\documentclass[10pt, conference,onecolumn]{IEEEtran}  

\usepackage{times}
\usepackage{hyperref} 
\usepackage{amsmath}
\usepackage{amssymb}
\usepackage{amsthm}
\usepackage{bbm}
\usepackage{dsfont}

\usepackage{epstopdf} 
\usepackage{graphicx}
\usepackage{url}

\usepackage[normalem]{ulem}

\usepackage{color}

\newtheorem{lemma}{Lemma}

\theoremstyle{definition}

\theoremstyle{remark}

\newcommand{\mynewline}{\mbox{}\\}
\newcommand{\E}[1]{\mathbb{E} \left( #1 \right)}

\newcommand{\Es}[2]{\mathbb{E}_{#1} \left( #2 \right)}

\newcommand{\BRA}[1]{\left( #1 \right)}

\newcommand{\BRAs}[1]{\left\{ #1 \right \}}
\newcommand{\PR}[1]{Pr\left\{ #1 \right\}}
\newcommand{\PRs}[2]{Pr_{#1}\left\{ #2 \right\}}

\newcommand{\ie}{{\emph{i.e.}}}
\newcommand{\eg}{{\emph{e.g.}}}

\newcommand{\renyi}{ R{\'e}nyi }

\newcommand{\cramer}{ Cram{\'e}r }
\newcommand{\beress}{ Berry-Ess{\'e}en }
\newcommand{\elisgar}{ Ellis--G\"{a}rtner }

\newcommand{\Ind}[1]{ \mathds{1}_{\BRAs{#1}} }

\newcommand{\MIN}[1]{ \smash{\displaystyle\min_{#1}} }

\newcommand{\toinf} {\underset{n \rightarrow \infty}{\rightarrow}}

\newif\ifFullProofs
\FullProofsfalse 

\newif\ifTwoCols
\TwoColstrue 

\begin{document}

\title{Variational formulas for the power of the binary hypothesis testing problem with applications}

\author{Nir~Elkayam ~~~~~~~Meir~Feder \\
        Department of Electrical Engineering - Systems\\
        Tel-Aviv University, Israel \\
        Email: nirelkayam@post.tau.ac.il, meir@eng.tau.ac.il}

\maketitle

\subsection*{\centering Abstract}
\textit{
Two variational formulas for the power of the binary hypothesis testing problem are derived. The first is given as the Legendre transform of a certain function and the second, induced from the first, is given in terms of the Cumulative Distribution Function (CDF) of the log-likelihood ratio. One application of the first formula is an upper bound on the power of the binary hypothesis testing problem in terms of the \renyi divergence. The second formula provide a general framework for proving asymptotic and non-asymptotic expressions for the power of the test utilizing corresponding expressions for the CDF of the log-likelihood. The framework is demonstrated in the central limit regime (\ie, for non-vanishing type I error) and in the large deviations regime.}

\section{Introduction}
A classical problem in statistics and information theory is the binary hypothesis testing problem where two distributions, $P$ and $Q$, are given. For each \emph{test}, we have two types of errors, namely the \emph{miss-detection} (type I) and the \emph{false-alarm} (type II) errors. According to the Neyman-Pearson lemma, the optimal test is based on thresholding the likelihood ratio between $P$ and $Q$.
The behavior of the optimal tradeoff between the two types of errors has been studied both in the asymptotic and non-asymptotic regimes, and in the central limit regime and large deviation regime.
Knowledge of the optimal tradeoff turns out to be useful for recent studies in finite block-length information theory, e.g., in channel coding
\cite[Section III.E]{polyanskiy2010channel},\cite{ElkayamITW2015} in data compression \cite{kostina2012fixed} and more.

Consider two probability measures $P$ and $Q$ on a sample space $W$\footnote{throughout this paper we assume discrete probability measure}. This paper provides two variational formulas for the power (or the optimal tradeoff) of the binary hypothesis problem between $P$ and $Q$. The first is given as the Legandere transform of the convex function (of $\lambda$):
$$ f(\lambda) = \lambda-\sum_{w\in W} \min\BRA{Q(w),\lambda P(w)}$$
and the second, derived from the first, is given as a function of the CDF of the log-likelihood ratio between $P$ and $Q$ with respect to $P$.

We use the first formula to derive a general upper bound on the power of any binary hypothesis testing problem in terms of the \renyi divergence between $P$ and $Q$. The second formula leads to a general framework for proving asymptotic and non-asymptotic bounds on the power of the binary hypothesis testing problem by plugging-in any approximation for the CDF of the log-likelihood ratio. The error term in the CDF approximation leads to a corresponding error term in the power of the binary hypothesis problem.
Specifically, by using the \beress theorem we get an approximation of the CDF up to an additive term, which results in an additive error term in the corresponding type I error. By using a large deviation approximation of the CDF, we get an approximation within a multiplicative term which results in a multiplicative error term in the corresponding type I error.

\section{Variational formula's for the Binary Hypothesis testing problem}\label{Sec:BinaryHypLemma}

Recall some general (and standard) definitions for the optimal performance of a binary hypothesis testing between two probability measures $P$ and $Q$ on $W$:
\begin{equation}\label{binary_hypothsis:beta}
  \beta_{\alpha}\BRA{P,Q} = \MIN{\substack{P_{Z|W} :\\ \sum_{w\in W}P(w)P_{Z|W}(1|w) \geq \alpha} } \sum_{w\in W}Q(w)P_{Z|W}(1|w),
\end{equation}
where $P_{Z|W}:W \rightarrow \BRAs{0,1}$ is any randomized test. The minimum is guaranteed to be achieved by the
Neyman--Pearson lemma. Thus, $\beta_{\alpha}\BRA{P,Q}$ gives the minimum probability of error under hypothesis $Q$ given that the probability of error under hypothesis $P$ is not larger than $1-\alpha$. The quantity $\beta$ denotes  the \textbf{power} of the test at \textbf{significance level} $1-\alpha$.

Recall that the optimal test is:
  \begin{equation*}
    P_{Z|W} = \Ind{ \frac{Q(w)}{P(w)} < \lambda} + \delta\cdot\Ind{ \frac{Q(w)}{P(w)} = \lambda},
  \end{equation*}
  where $\lambda, \delta$ are tuned so that $\sum_{w\in W}P(w)P_{Z|W}(1|w) = \alpha$

\begin{lemma}\label{Lemma:BinaryHyp}
The following variational formula holds:
\begin{equation}\label{Formula:Beta:Sup}
  \beta_{\alpha}\BRA{P,Q} = \max_{\lambda}\BRA{\sum_{w\in W} \min\BRA{Q(w),\lambda P(w)} - \lambda\BRA{1-\alpha}}.
\end{equation}
Moreover,
\begin{equation}\label{Formula:Beta}
  \beta_{\alpha}\BRA{P,Q} = \sum_{w\in W} \min\BRA{Q(w),\lambda P(w)} - \lambda\BRA{1-\alpha}
\end{equation}
If and only if:
\begin{equation}\label{Formula:OptimalLambda}
  P\BRAs{w:\frac{Q(w)}{P(w)} < \lambda} \leq \alpha \leq P\BRAs{w:\frac{Q(w)}{P(w)} \leq \lambda}
\end{equation}
\end{lemma}

The next lemma presents another variational formula, where the power of the binary hypothesis testing problem is expressed in terms of the CDF of the log-likelihood ratio.

\begin{lemma}\label{Lemma:BinaryHyp:2}
Let $F(z) = \PRs{P}{\log {P(w) \over Q(w)} \leq z}$ denote the CDF of the log-likelihood ratio with respect to the distribution $P$.
Then:
\begin{align}
  \beta_{1-\alpha}\BRA{P,Q} &= \max_{R}\BRA{\int_R^{\infty}F(z)e^{-z}dz  - e^{-R}\alpha} \label{Formula:Beta:CDF}
\end{align}

Moreover,
\begin{equation}\label{Formula:Beta:2}
  \beta_{1-\alpha}\BRA{P,Q} = \int_R^{\infty}F(z)e^{-z}dz  - e^{-R}\alpha
\end{equation}
If and only if:
\begin{equation}\label{OptimalR}
  P\BRAs{w:\log\frac{P(w)}{Q(w)} < R} \leq \alpha \leq P\BRAs{w:\log\frac{P(w)}{Q(w)} \leq R}
\end{equation}
\end{lemma}
The proof of Lemmas 1 and 2 appears in Appendix \ref{App:BinaryHypLemma}.

The following equality may facilitate the expression in \eqref{Formula:Beta:CDF}
\begin{align}
  &\int_R^{\infty}F(z)e^{-z}dz \notag \\
  &= F(R)e^{-R}+\Es{P}{e^{-\log\frac{P(z)}{Q(z)}}\Ind{\log\frac{P(z)}{Q(z)} > R}} \label{Formula:ChangeOfMeasure}
\end{align}
This relation follows by using integration in parts and the Riemann-Stieltjes integral, and it holds for both discrete and continuous distributions. Its proof is straight-forward and omitted due to lack of space.

It is interesting to note that for $R$ such that $F(R)=\alpha$, \eqref{Formula:Beta:2}, \eqref{OptimalR} and \eqref{Formula:ChangeOfMeasure} imply:
  \begin{align}
    \beta_{1-\alpha}\BRA{P,Q} &= \Es{P}{e^{-\log\frac{P(z)}{Q(z)}}\Ind{\log\frac{P(z)}{Q(z)} > R}} \notag \\
    &= \Es{Q}{\Ind{\log\frac{P(z)}{Q(z)} > R}} \notag \\
    &= Q\BRA{w: \log\frac{P(z)}{Q(z)} > R} \label{Beta:ChangeOfMeasure}
  \end{align}
which gives an intuition and serves as a sanity check for the proposed formulas.

\section{Applications}

\subsection{An upper bound on $\beta_{1-\alpha}\BRA{P,Q}$ in terms of the \renyi divergence}

Let:
\begin{equation*}
  g_s\BRA{P,Q} \triangleq \log\BRA{\sum_w P(w)^s Q(w)^{1-s}}
\end{equation*}

For $0 \leq s \leq 1$:
\begin{equation*}
  \min\BRA{Q(w),\lambda P(w)} \leq Q(w)^{1-s}\BRA{\lambda P(w)}^s = \lambda^s e^{g_s}
\end{equation*}
and:
\begin{align*}
\beta_{1-\alpha}\BRA{P,Q} &= \max_{\lambda}\BRA{\sum_{w\in W} \min\BRA{Q(w),\lambda P(w)} - \lambda\alpha} \\
 &\leq \max_{\lambda}\BRA{\lambda^s e^{g_s} - \lambda\alpha} \\
 &= e^{g_s} \max_{\lambda}\BRA{\lambda^s  - \lambda{\alpha e^{-g_s}}}
\end{align*}
Taking the $\log$:
\begin{align*}
 \log\beta_{1-\alpha}\BRA{P,Q} &\leq g_s + \log\BRA{\max_{\lambda}\BRA{\lambda^s  - \lambda{\alpha e^{-g_s}}}} \\
 &\overset{(a)}{\leq} g_s + \frac{s\log(\alpha e^{-g_s})+h_b(s)}{s-1} \\
 &= -{1 \over s-1}g_s + {s \over s-1}\log(\alpha)+\frac{h_b(s)}{s-1}
\end{align*}
where (a) follows by an elementary calculus\footnote{here $h_b(s)$ is the standard binary entropy function in nats,
 ie: $h_b(s)=-slog(s)-(1-s)log(1-s)$}. Note that ${1 \over s-1}g_s = D_s(P||Q)$ is the \renyi divergence. Taking $\alpha=e^{-r}$, and optimize for $s$:
\begin{equation}\label{GeneralExponentialBound}
  \log\beta_{1-e^{-r}}\BRA{P,Q} \leq \inf_{0 \leq s \leq 1}-D_s(P||Q) - {s \over s-1}r+\frac{h_b(s)}{s-1}
\end{equation}
we get the bound on the error exponents.

\subsection{Normal approximations}
Let $G(z)$ be another CDF (\eg\ Gaussian), approximating $F(z)$ with an additive approximation error, \ie\, $$G(z)-d_l\leq F(z) \leq G(z)+d_h$$ for all $z$. Denote by
\begin{equation}\label{Beta:CDF}
  \beta_{1-\alpha}\BRA{F} \triangleq \max_{R}\BRA{\int_R^{\infty}F(z)e^{-z}dz - e^{-R}\alpha}
\end{equation}
the power of the binary hypothesis testing in terms of the CDF of the log-likelihood ratio.
Then:
\begin{align*}
  &\beta_{1-\alpha}\BRA{F} = \beta_{1-\alpha}\BRA{P,Q} \\
  &=\max_{R}\BRA{\int_R^{\infty}F(z)e^{-z}dz  - e^{-R}\alpha} \\
  &\leq \max_{R}\BRA{\int_R^{\infty}\BRA{G(z)+d_h}e^{-z}dz  - e^{-R}\alpha} \\
  &= \max_{R}\BRA{\int_R^{\infty}G(z)e^{-z}dz +d_h e^{-R} - e^{-R}\alpha} \\
  &= \max_{R}\BRA{\int_R^{\infty}G(z)e^{-z}dz - e^{-R}\BRA{\alpha-d_h}} \\
  &= \beta_{1-\alpha+d_h}\BRA{G}
\end{align*}
And similarly:
\begin{align*}
  \beta_{1-\alpha}\BRA{P,Q} \geq \beta_{1-\alpha-d_l}\BRA{G}
\end{align*}

Let $G$ be now a CDF of a Gaussian distribution approximating $F$. In this case $\beta_{1-\alpha}\BRA{G}$ can be evaluated explicitly.
Specifically, let $L(w)=\log\frac{P(w)}{Q(w)}$ and assume $L \sim \mathcal{N}\BRA{D,V}$ under $P$, \ie, the likelihood ratio is distributed normally, then:
\begin{align*}
  G(z) &=  \PRs{P}{w:\log {P(w) \over Q(w)} \leq z} \\
       &= \PRs{P}{w:L(w) \leq z} \\
       &= \PRs{P}{w:\frac{L(w)-D}{\sqrt{V}} \leq \frac{z-D}{\sqrt{V}} } \\
       &= \Phi\BRA{\frac{z-D}{\sqrt{V}}}
\end{align*}
where $\Phi(x)={1 \over \sqrt{2\pi}}\int_{-\infty}^x e^{-t^2/2}dt$. The optimal $R=R_\alpha$ is given by equation \eqref{OptimalR}: $\Phi\BRA{\frac{R_\alpha-D}{\sqrt{V}}} = \alpha$, \ie
\begin{equation}\label{Gaussian:OptimalR}
  R_{\alpha} = D+\sqrt{V}\Phi^{-1}(\alpha)
\end{equation}
We can use \eqref{Beta:ChangeOfMeasure}:
\begin{align}
  &\E{e^{-\log\frac{P(z)}{Q(z)}}\Ind{\log\frac{P(z)}{Q(z)} > R}} \notag \\
  &= \int_{R}^{\infty}e^{-z}\frac{1}{\sqrt{2\pi V}}e^{-\frac{(z-D)^2}{2V}}dz \notag \\
  &\overset{(a)}{=} e^{-D}\int_{R-D}^{\infty}e^{-z}\frac{1}{\sqrt{2\pi V}}e^{-\frac{z^2}{2V}}dz \notag \\
  &= e^{-D}\cdot\int_{R-D}^{\infty}\frac{1}{\sqrt{2\pi V}}e^{-{1\over 2}\frac{(z+V)^2-V^2}{V}}dz \notag \\
  &= e^{-D+V/2}\cdot\int_{R-D+V}^{\infty}\frac{1}{\sqrt{2\pi V}}e^{-{1\over 2}\frac{z^2}{V}}   dz \notag \\
  &= e^{-D+V/2}\cdot \Phi^C\BRA{\frac{R-D+V}{\sqrt{V}}} \label{Beta:ChangeOfMeasure:Gaussian}
\end{align}
where $\Phi^C(x)=1-\Phi(x)$\footnote{$\Phi^C$ is the $Q$ function of the Gaussian distribution.}.
Plugging \eqref{Gaussian:OptimalR} into \eqref{Beta:ChangeOfMeasure:Gaussian} we get:
\begin{align*}
  &\E{e^{-\log\frac{P(z)}{Q(z)}}\Ind{\log\frac{P(z)}{Q(z)} > R_{\alpha}}} \\
  &= e^{-D+V/2}\cdot \Phi^C\BRA{\Phi^{-1}(\alpha)+\sqrt{V}}
\end{align*}
Using the following approximation of $\Phi^C(t)$, \cite[Formula 7.1.13]{abramowitz1964handbook}:
\begin{equation}\label{Q_approx}
  \sqrt{2 \over \pi}\frac{e^{-t^2/2}}{t+\sqrt{t^2+4}} \leq \Phi^C(t) \leq \sqrt{2 \over \pi}\frac{e^{-t^2/2}}{t+\sqrt{t^2+8/\pi}}
\end{equation}
we get for $t>>0$ (\eg\, when $V >> 0$):
$$ \Phi^C(t) \approx \sqrt{2 \over \pi}\frac{e^{-t^2/2}}{2t} = \sqrt{1 \over 2\pi}\frac{e^{-t^2/2}}{t}$$
and so:
\begin{align*}
&e^{-D+V/2}\cdot \Phi^C \BRA{\Phi^{-1}(\alpha)+\sqrt{V}} \\
&\approx e^{-D+V/2}\cdot \sqrt{1 \over 2\pi}\frac{e^{-{\BRA{\Phi^{-1}(\alpha)+\sqrt{V}}^2\over 2}}}{\Phi^{-1}(\alpha)+\sqrt{V}} \\
&\approx {e^{-D-\sqrt{V}\Phi^{-1}(\alpha)}\over \sqrt{V}}
\end{align*}
where $\approx$ denote equality up to a multiplicative constant term.

For the block memoryless hypothesis testing problem, the CDF of the log-likelihood ratio can be approximated via the \beress theorem \cite[Theorem 1.2]{tan2014asymptotic}. In this paper we omit the details and refer the interested reader to the monograph \cite{tan2014asymptotic}. Note also that the results above strengthen the existing results on the finite block length source coding problem as it is an instance of the binary hypothesis testing problem \cite[3.2]{tan2014asymptotic}.

\subsection{Large deviation regime}
In this section we outline the application of the variational formula \eqref{Formula:Beta:CDF} to the analysis of the large deviation regime of the binary hypothesis testing problem. This might simplify the proof of the formula for the error exponent of the power of the binary hypothesis testing problem with general sources \cite[Theorem 2.1]{sun2000hypothesis}.

The setting is as follow. A sequence of distributions $P_n, Q_n$ is given and the limit (if exists) of:
\begin{align}
  E_n(r) &= -{1 \over n}\log\beta_{1-e^{-nr}}\BRA{P_n,Q_n}
\end{align}
is of interest. Let:
\begin{equation}
  F_n(z) = \PRs{P_n}{w:\log{P_n(w) \over Q_n(w)} \leq z}
\end{equation}
and assume that we have a large deviation approximation of $F_n(z)$, \ie:
\begin{equation}\label{Beta:CDF:LargeDev:Estimate}
  e^{-nE_1(z)-n\delta^l_n} \leq F_n(nz) \leq e^{-nE_1(z)+n\delta^h_n}
\end{equation}
The sequences $\delta^l_n, \delta^h_n \toinf 0$ represent the multiplicative approximation error. We refer the reader to \cite{dembo2009large} for the theorems of \cramer and \elisgar showing when such approximations exist.

Let:
\begin{equation}\label{Beta:CDF:LargeDev}
  f_n(r,R) = n\int_{R}^{\infty}e^{-nE_1(z)}e^{-nz}dz - e^{-n(R+r)}
\end{equation}

Then:
\begin{align}
&\beta_{1-e^{-nr}}\BRA{P_n,Q_n} \notag \\
&= \max_{R}\BRA{\int_{nR}^{\infty}F_n(z)e^{-z}dz  - e^{-nR}e^{-nr}} \notag \\
&\overset{(a)}{=} \max_{R}\BRA{n\int_{R}^{\infty}F_n(nz)e^{-nz}dz  - e^{-n(R+r)}} \notag \\
&\leq \max_{R}\BRA{n\int_{R}^{\infty}e^{-nE_1(z)-nz+n\delta^h_n}dz  - e^{-n(R+r)}} \notag \\
&= e^{n\delta^h_n}\max_{R}\BRA{n\int_{R}^{\infty}e^{-nE_1(z)-nz}dz  - e^{-n(R+r+\delta^h_n)}} \notag \\
&= e^{n\delta^h_n}\max_{R}\BRA{f_n(r+\delta^h_n,R)} \label{Beta:LargeDev:UpperBound}
\end{align}
where (a) comes from a change of variables $z=nz'$. The lower bound also follows:
\begin{align}\label{Beta:LargeDev:LowerBound}
&\beta_{1-e^{-nr}}\BRA{P_n,Q_n} \geq e^{-n\delta^l_n}\max_{R}\BRA{f_n(r-\delta^l_n,R)}
\end{align}
Let:
\begin{equation}\label{Def:E2}
  E_{2,n}(r) = - {1 \over n}\log\max_{R}\BRA{f_n(r,R)}
\end{equation}
Combining \eqref{Def:E2}, \eqref{Beta:LargeDev:UpperBound} and \eqref{Beta:LargeDev:LowerBound}:
\begin{equation}
  -\delta^h_n + E_{2,n}(r+\delta^h_n) \leq E_n(r) \leq \delta^l_n + E_{2,n}(r-\delta^l_n)
\end{equation}
So the analysis of $E_{2,n}(r)$ is of interest. Differentiating  $f_n(r,R)$ with respect to $R$:

\begin{align*}
  \frac{d}{dR}f_n(r,R) &\overset{(a)}{=} -ne^{-nE_1(R)}e^{-nR} +ne^{-n(R+r)}  \\
  &= ne^{-nR}\BRA{e^{-nr}-e^{-nE(R)}  }
\end{align*}
where (a) follows by differentiation under the integral sign\footnote{ $\frac{d}{dx}\BRA{\int_{a(x)}^{b(x)} f(x,t)dt} = f(x,b(x))b'(x) - f(x,a(x))a'(x) + \int_{a(x)}^{b(x)} f_x(x,t)dt $}\cite{flanders1973differentiation}.
The optimal $R$ satisfy:
\begin{equation}\label{LargeDeviation:OptimalR}
  r = E(R)
\end{equation}

The asymptotic analysis of $E_{2,n}(r)$ can be carried out using the Laplace method of integration which leads to the asymptotic behavior of $E_n(r)$. We omit the details due to space limitations.

\section{Summary}
Two variational formulas for the power of the binary hypothesis testing problem were derived. The formulas can provide tighter bounds on the power of the test, \eg, in terms of the \renyi divergence. Furthermore, a framework for approximating the power of the optimal test is proposed. Any approximation of the CDF of the log-likelihood ratio will result in an approximation to power of the binary hypothesis testing problem. The approximating CDF should be simple enough to allow the calculation of an approximated (or exact) power of the approximating problem. Specifically, we have shown that for the Gaussian approximation, the exact power of the test can be calculated. In the large deviation regime, the power of the binary hypothesis problem can also be approximated using the proposed framework by utilizing the Laplace's method of integration.

\appendices

\section{Proofs}\label{App:BinaryHypLemma}

\begin{proof}[Proof of lemma \ref{Lemma:BinaryHyp}:]
\mynewline
\underline{Proof of \eqref{Formula:Beta}}

  Let $\lambda, \delta$ be the thresholds for the optimal test, and let:
  $$A = \BRAs{w:\frac{Q(w)}{P(w)} < \lambda} $$
  $$B = \BRAs{w:\frac{Q(w)}{P(w)} = \lambda} $$
  Then:
  \begin{equation}\label{Def:Alpha}
    \alpha = P(A)+\delta P(B)
  \end{equation}
  And:
  \begin{equation}\label{Def:Beta}
    \beta = Q(A)+\delta Q(B)
  \end{equation}
  Multiply \eqref{Def:Alpha} by $\lambda$, subtract \eqref{Def:Beta} and use $Q(B) = \lambda P(B)$:
  $$ \beta-\lambda\alpha = Q(A)-\lambda P(A)$$

  On the other hand:
  \begin{align*}
    \sum_{w\in W} \min\BRA{Q(w),\lambda P(w)} &= \sum_{w \in A}Q(w) + \sum_{w \in A^c}\lambda P(w)\\
    &= Q(A) + \lambda (1-P(A)) \\
    &= Q(A) - \lambda P(A)+\lambda \\
    &=\beta-\lambda\alpha+\lambda
  \end{align*}
  Thus:
  $$ \beta = \sum_{w\in W} \min\BRA{Q(w),\lambda P(w)} -\lambda(1-\alpha)$$

\underline{Proof of the sup formula (smaller $\lambda$):}

Note that the optimal $\lambda$ satisfies the following:
\begin{equation}\label{OptimalLambda}
  P\BRAs{w:\frac{Q(w)}{P(w)} \geq \lambda} \geq 1-\alpha \geq P\BRAs{w:\frac{Q(w)}{P(w)} > \lambda}
\end{equation}
  Let $\lambda_1 < \lambda$:
\begin{align*}
  & \sum_{w\in W} \min\BRA{Q(w),\lambda_1 P(w)} - \sum_{w\in W} \min\BRA{Q(w),\lambda P(w)} \\
         &= \sum_{w\in W:\lambda_1 P(w) < Q(w) < \lambda P(w)} \BRA{\lambda_1 P(w) - Q(w)} \\
         &+\BRA{\lambda_1  - \lambda }\sum_{w\in W:\lambda P(w) \leq Q(w) } P(w) \\
         &\overset{(a)}{\leq} \BRA{\lambda_1  - \lambda }\sum_{w\in W:\lambda P(w) \leq Q(w) } P(w) \\
         &= \BRA{\lambda_1  - \lambda }P\BRAs{w: \frac{Q(w)}{P(w)} \geq \lambda } \\
         &\overset{(b)}{\leq} \BRA{\lambda_1  - \lambda }(1-\alpha)
\end{align*}
where (a) follow from: $\lambda_1 P(w) - Q(w)< 0$, (b) follow from $\lambda_1  - \lambda < 0$ and $P\BRAs{w: \frac{Q(w)}{P(w)} \geq \lambda } \geq 1-\alpha$.
Rearranging the terms:
\begin{align*}
  &\sum_{w\in W} \min\BRA{Q(w),\lambda_1 P(w)} -\lambda_1(1-\alpha) \\
  &\leq \sum_{w\in W} \min\BRA{Q(w),\lambda P(w)}-\lambda(1-\alpha)
\end{align*}
If $\lambda_1$ does not satisfy the condition \eqref{Formula:OptimalLambda}, then:
\begin{itemize}
  \item If $P\BRAs{w:\frac{Q(w)}{P(w)} \leq \lambda_1} < P\BRAs{w:\frac{Q(w)}{P(w)} < \lambda}$, then we are finished because there exist $w_0$ with $P(w_0) > 0$, $\frac{Q(w_0)}{P(w_0)} < \lambda$, and $\frac{Q(w_0)}{P(w_0)} > \lambda_1$, which gives strict inequality in (a) above.
  \item If $P\BRAs{w:\frac{Q(w)}{P(w)} \leq \lambda_1} = P\BRAs{w:\frac{Q(w)}{P(w)} < \lambda}$ then $P\BRAs{w:\frac{Q(w)}{P(w)} \leq \lambda_1} < \alpha$ and we have strict inequality $P\BRAs{w:\frac{Q(w)}{P(w)} < \lambda} < \alpha$, which leads to a strict inequality in (b) above.
\end{itemize}

\underline{Proof of the sup formula (greater $\lambda$):}

For $\lambda_1 > \lambda$ we have:

\begin{align*}
  &\sum_{w\in W} \min\BRA{Q(w),\lambda_1 P(w)} \\
  &=    Q\BRAs{w:Q(w) < \lambda P(w)}\\
  &+Q\BRAs{w:\lambda P(w) \leq Q(w) \leq \lambda_1 P(w)}\\
  &+\lambda_1 P\BRAs{w:Q(w) > \lambda_1 P(w)} \\
  &\overset{(a)}{\leq} Q\BRAs{w:Q(w) < \lambda P(w)}\\
  &+\lambda_1 P\BRAs{w:\lambda P(w) \leq Q(w) \leq \lambda_1 P(w)}\\
  &+\lambda_1 P\BRAs{w:Q(w) > \lambda_1 P(w)} \\
  &=    Q\BRAs{w:Q(w) < \lambda P(w)}+\lambda_1 P\BRAs{w:Q(w) \geq \lambda P(w)}
\end{align*}
where (a) follow upper bounding $Q(w)$ with $\lambda_1 P(w)$.
\begin{align*}
  & \sum_{w\in W} \min\BRA{Q(w),\lambda_1 P(w)} - \sum_{w\in W} \min\BRA{Q(w),\lambda P(w)} \\
  &\leq Q\BRAs{w:Q(w) < \lambda P(w)}\\
  &+\lambda_1 P\BRAs{w:Q(w) \geq \lambda P(w)}\\
  &-Q\BRAs{w:Q(w) < \lambda P(w)}-\lambda P\BRAs{w:Q(w) \geq \lambda P(w)} \\
  &= (\lambda_1-\lambda) P\BRAs{w:Q(w) \geq \lambda P(w)} \\
  &\leq \BRA{\lambda_1-\lambda} (1-\alpha)
\end{align*}
Since $\lambda_1-\lambda > 0$ and $P\BRAs{w:Q(w) \geq \lambda P(w)} \leq 1-\alpha$, we have:
\begin{align*}
  &\sum_{w\in W} \min\BRA{Q(w),\lambda_1 P(w)} -\lambda_1(1-\alpha) \\
  &\leq \sum_{w\in W} \min\BRA{Q(w),\lambda P(w)}-\lambda(1-\alpha)
\end{align*}

If $\lambda_1$ does not satisfy the condition \eqref{Formula:OptimalLambda}, then $P\BRAs{w:\frac{Q(w)}{P(w)} < \lambda} < P\BRAs{w:\frac{Q(w)}{P(w)} < \lambda_1}$ and we are finished because there exist $w_0$ with $P(w_0) > 0$, $\frac{Q(w_0)}{P(w_0)} \geq \lambda$, and $\frac{Q(w_0)}{P(w_0)} < \lambda_1$, which gives strict inequality in (a) above.

\end{proof}


\begin{proof}[Proof of lemma \ref{Lemma:BinaryHyp:2}:]

\begin{align*}
  \sum_{w\in W} \min\BRA{Q(w),\lambda P(w)} &= \sum_{w\in W} \lambda P(w)\min\BRA{{Q(w) \over \lambda P(w)},1} \\
  &= \lambda\Es{P}{\min\BRA{{Q(w) \over \lambda P(w)},1}}
\end{align*}
where we take $\frac{c}{0}=\infty$ for $c \geq 0$ and $0\cdot\infty =0$ in order to handle the case $P(w)=0$ as well.
In \cite[Lemma 2]{ElkayamITW2015} we proved that for any positive random variable $U$:
\begin{equation*}
  \E{\min\BRA{e^R \cdot U,1}}= e^{R}\cdot\int_R^{\infty}F(z)e^{-z}dz
\end{equation*}
where:
\begin{equation*}
  F(z) = \PR{-\log U \leq z}
\end{equation*}
Here we define $U$ as $U(w)={Q(w) \over P(w)}$ and we get:
\begin{equation*}
  \Es{P}{\min\BRA{e^R{Q \over P},1}}= e^{R}\cdot\int_R^{\infty}F(z)e^{-z}dz
\end{equation*}
where:
\begin{align*}
  F(z) &= \PRs{P}{z:-\log {Q(w) \over P(w)} \leq z} \\
  &= \PRs{P}{z:\log {P(w) \over Q(w)} \leq z}
\end{align*}

Taking $\lambda=e^{-R}$ we get:
\begin{align*}
  \sum_{w\in W} \min\BRA{Q(w),\lambda P(w)} &= \lambda\Es{P}{\min\BRA{{Q \over \lambda P},1}}\\
  &= e^{-R}\cdot\Es{P}{\min\BRA{e^{R}{Q \over P},1}}\\
  &= \int_R^{\infty}F(z)e^{-z}dz
\end{align*}

Plugging back in \eqref{Formula:Beta} and writing for $1-\alpha$ instead of $\alpha$:
\begin{align*}
  \beta_{1-\alpha}\BRA{P,Q} &= \max_{\lambda}\BRA{\sum_{w\in W} \min\BRA{Q(w),\lambda P(w)} - \lambda(\alpha} \\
  &= \max_{R}\BRA{\int_R^{\infty}F(z)e^{-z}dz  - e^{-R}\alpha}
\end{align*}
The optimal $\lambda$ satisfy \eqref{Formula:OptimalLambda}. Rewriting the condition in terms of $R$ after some algebra we get:
%
%
\begin{equation}
  P\BRAs{w:\log\frac{P(w)}{Q(w)} < R} \leq \alpha \leq P\BRAs{w:\log\frac{P(w)}{Q(w)} \leq R}
\end{equation}

\end{proof}

\ifFullProofs
\else

\fi

\bibliographystyle{IEEEtran}
\bibliography{bib}

\end{document}